\documentclass[11pt]{article} \usepackage[normal]{optional}

\usepackage{amsmath,amssymb,amsfonts,amsthm}
\opt{normal}{\usepackage[text={6.5in,9in},centering]{geometry} \geometry{letterpaper}}
\usepackage{graphicx}
\usepackage{ifpdf}
\usepackage{cite}
\usepackage{comment}
\usepackage[textsize=scriptsize]{todonotes} 

\opt{normal}{\usepackage{subfigure}}

\usepackage{sidecap,caption}

\newcommand{\figuresize}{\small}

\usepackage{commands-tam}
\opt{normal}{\usepackage{numbering-tam}}

\newcommand{\RR}{\mathbb{R}}
\newcommand{\NN}{\mathbb{N}}
\newcommand{\ZZ}{\mathbb{Z}}

\newcommand{\ignore}[1]{}

\newcommand{\vv}{\vec{v}}

\ifpdf
  \usepackage[pdftex]{epsfig}
  \opt{normal}{\usepackage[pdfstartview={XYZ null null 1.00},pdfpagemode=None,colorlinks=true,citecolor=blue,linkcolor=blue]{hyperref}}
  \opt{socg}{\usepackage{hyperref}}
\else
    \usepackage[dvips]{epsfig}
    \opt{normal}{\usepackage[dvips,colorlinks=true,citecolor=blue,linkcolor=blue]{hyperref}}
    \opt{socg}{\usepackage{hyperref}}
\fi

\begin{document}

\title{Pattern overlap implies runaway growth in hierarchical tile systems}

\opt{normal}{
  \author{
    Ho-Lin Chen\thanks{National Taiwan University, Taipei, Taiwan, {\tt holinc@gmail.com}. Supported by NSC grant number 101-2221-E-002-122-MY3.}
    \and
    David Doty\thanks{California Institute of Technology, Pasadena, CA, USA, {\tt ddoty@caltech.edu}. The author was supported by the Molecular Programming Project under NSF grants 0832824 and 1317694 and by NSF grants CCF-1219274 and CCF-1162589, and by a Computing Innovation Fellowship under NSF grant 1019343.}
    \and
    J\'{a}n Ma\v{n}uch\thanks{University of British Columbia,
      Vancouver BC, Canada and Simon Fraser University, Burnaby BC,
      Canada, {\tt jmanuch@sfu.ca}. Supported by NSERC Discovery grant.}
    \and
    Arash Rafiey\thanks{Simon Fraser University, Burnaby, BC, Canada, {\tt arashr@sfu.ca}}
    \and
    Ladislav Stacho\thanks{Simon Fraser University, Burnaby, BC,
      Canada, {\tt lstacho@sfu.ca}. Supported by NSERC Discovery grant.}
  }
}

\opt{normal}{\date{}}

\maketitle

\begin{abstract}
  We show that in the hierarchical tile assembly model, if there is a producible assembly that overlaps a nontrivial translation of itself consistently (i.e.,~the pattern of tile types in the overlap region is identical in both translations), then arbitrarily large assemblies are producible.
  The significance of this result is that tile systems intended to controllably produce finite structures must avoid pattern repetition in their producible assemblies that would lead to such overlap.

  This answers an open question of Chen and Doty (\emph{SODA 2012}), who showed that so-called ``partial-order'' systems producing a unique finite assembly {\bf and} avoiding such overlaps must require time linear in the assembly diameter.
  An application of our main result is that any system producing a unique finite assembly is automatically guaranteed to avoid such overlaps, simplifying the hypothesis of Chen and Doty's main theorem.
\end{abstract}



\thispagestyle{empty}\newpage\setcounter{page}{1}

\section{Introduction}
\label{sec-intro}

Winfree's abstract Tile Assembly Model (aTAM)~\cite{Winfree98simulationsof} is a model of crystal growth through cooperative binding of square-like monomers called \emph{tiles}, implemented experimentally (for the current time) by DNA~\cite{WinLiuWenSee98, BarSchRotWin09}.
It models the potentially algorithmic capabilities of tiles that are designed to bind if and only if the total strength of attachment (summed over all binding sites, called \emph{glues} on the tile) is at least a threshold $\tau$, sometimes called the \emph{temperature}.
When glue strengths are integers and $\tau=2$, two strength 1 glues must cooperate to bind the tile to a growing assembly.
Two assumptions are key: 1) growth starts from a single \emph{seed} tile type, and 2) only individual tiles bind to an assembly.
We refer to this model as the \emph{seeded aTAM}.

While violations of these assumptions are often viewed as errors in implementation of the seeded aTAM~\cite{SchWin07, SchWin09}, relaxing them results in a different model with its own programmable abilities.
In the \emph{hierarchical} (a.k.a. \emph{multiple tile}~\cite{AGKS05}, \emph{polyomino}~\cite{Winfree06, Luhrs10}, \emph{two-handed}~\cite{TwoHandsBetterThanOne, DotPatReiSchSum10, TwoHandedNotIntrinsicallyUniversalconf}) \emph{aTAM}, there is no seed tile, and an assembly is considered producible so long as two producible assemblies are able to attach to each other with strength at least $\tau$, with all individual tiles being considered as ``base case'' producible assemblies.
In either model, an assembly is considered \emph{terminal} if nothing can attach to it; viewing self-assembly as a computation, terminal assembly(ies) are often interpreted to be the output.
See~\cite{DotCACM,PatitzSurvey12} for an introduction to recent theoretical work using these models.

As with other models of computation, in general it is considerably more difficult to prove negative results (\emph{limitations} on what a tile system can do) than to prove positive results.
A common line of inquiry aimed at negative results in tile self-assembly concerns the notion of ``pumping'': showing that a single repetition of a certain group of tiles implies that the same group can be repeated indefinitely to form an infinite periodic structure.


The \emph{temperature-1} problem in the seeded model of tile assembly concerns the abilities of tile systems in which every positive-strength glue is sufficiently strong to bind two tiles.
It may seem ``obvious'' that if two tile types repeat in an assembly, then a segment of tiles connecting them could be repeated indefinitely (``pumped'') to produce an infinite periodic path (since, at temperature 1, each tile along the segment has sufficient strength for the next tile in the segment to attach).
However, this argument fails if the attempt to pump the segment ``crashes'' into an existing part of the assembly.
It is conjectured~\cite{jLSAT1} that only finite unions of periodic patterns (so-called \emph{semilinear} sets) can be assembled at temperature 1 in the seeded model, but despite considerable investigation~\cite{ManStaSto10,RotWin00,temp1noIU,reifsongtemp1}, the question remains open.
If true, temperature-1 \emph{hierarchical} tile systems would suffer a similar limitation, due to a formal connection between producible assemblies in the seeded and hierarchical models\cite[Lemma 4.1]{TwoHandsBetterThanOne}.
It \emph{has} been established, using pumping arguments, that temperature-1 seeded tile systems are unable to simulate the dynamics of certain temperature-2 systems~\cite{temp1noIU}.

Moving to temperature 2, both models gain power to assemble much more complex structures; both are able to simulate Turing machines, for instance.
In a certain sense, the hierarchical model is at least as powerful as the seeded model, since every seeded tile system can be simulated by a hierarchical tile system with a small ``resolution loss'': each tile in the seeded system is represented by a $5 \times 5$ block of tiles in the hierarchical system~\cite[Theorem 4.2]{TwoHandsBetterThanOne}.

From this perspective, the main theorem of this paper, a negative result on hierarchical tile assembly that does not apply to seeded tile assembly, is somewhat surprising.
We show that hierarchical systems, of \emph{any} temperature, are forced to admit a sort of infinite ``pumping'' behavior if a special kind of ``pattern repetition'' occurs.
More formally, suppose that a hierarchical tile system $\calT$ is able to produce an assembly $\alpha_0$ such that, for some nonzero vector $\vv$, the assembly $\alpha_1 = \alpha_0 + \vv$ (meaning $\alpha_0$ translated by $\vv$) intersects $\alpha_0$, but the two translations agree on every tile type in the intersection (they are \emph{consistent}).
It is known that this implies that the union $\alpha_0 \cup \alpha_1$ is producible as well~\cite[Theorem 5.1]{DotyProducibility}.
Our main theorem, Theorem~\ref{thm:main}, shows that this condition implies that $\calT$ can produce arbitrarily large assemblies, answering the main open question of~\cite{DotyProducibility}.

\begin{figure}[htb]
\begin{center}
  \includegraphics[width=\textwidth]{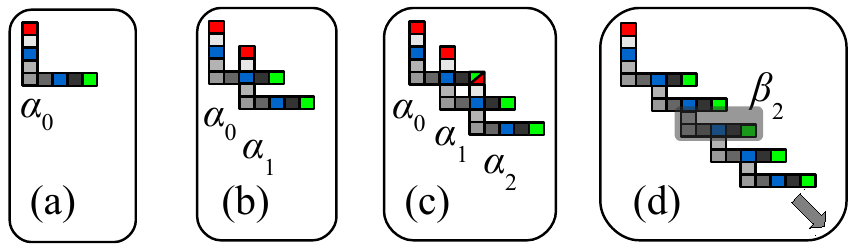}
  \vspace{-0.1in}
  \caption{\label{fig:pump} \figuresize
  Example of the main theorem of this paper.
  (a) A producible assembly $\alpha_0$. Gray tiles are all distinct types from each other, but red, green, and blue each represent one of three different tile types, so the two blue tiles are the same type.
  (b) By Theorem~\ref{thm-union-producible}, $\alpha_0 \cup \alpha_1$ is producible, where $\alpha_1 = \alpha_0 + (2,-2)$, because they overlap in only one position, and they both have the blue tile type there.
  (c) $\alpha_0$ and $\alpha_2$ both have a tile at the same position, but the types are different (green in the case of $\alpha_0$ and red in the case of $\alpha_2$).
  (d) However, a subassembly $\beta_i$ of each new $\alpha_i$ can grow, enough to allow the translated equivalent subassembly $\beta_{i+1}$ of $\alpha_{i+1}$ to grow from $\beta_i$, so an infinite structure is producible.
  }
\end{center}
\end{figure}

The assembly is not necessarily infinitely many translations of all of $\alpha_0$, since although $\alpha_0$ and $\alpha_1$ are consistent, which implies that $\alpha_1$ must be consistent with $\alpha_2 = \alpha_0 + 2 \vv$, it may be that $\alpha_0$ is not consistent with $\alpha_2$.
However, our proof shows that a subassembly $\beta_2$ of $\alpha_2$ can be assembled that is sufficient to grow another translated copy of $\beta_2$, so that the infinite producible assembly consists of infinitely many translations of $\beta_2$.
See Figure~\ref{fig:pump} for an example illustration.

An immediate application of this theorem is to strengthen a theorem of Chen and Doty~\cite{CheDot12}.
They asked whether every hierarchical tile system obeying a technical condition known as the \emph{partial order} property and producing a unique finite terminal assembly, also obeys the condition that no producible assembly is consistent with a translation of itself.
The significance of the latter condition is that the main theorem of~\cite{CheDot12} shows that systems satisfying the condition obey a time lower bound for assembly: they assemble their final structure in time $\Omega(d)$, where $d$ is the diameter of the final assembly.
Our main theorem implies that every system \emph{not} satisfying the condition must produce arbitrarily large assemblies and therefore cannot produce a unique finite terminal assembly.
Hence \emph{all} hierarchical partial order systems are constrained by this time lower bound, the same lower bound that applies to all seeded tile systems.
Thus hierarchical partial order systems, despite the ability to assemble many sub-assemblies of the final assembly in parallel, provably cannot exploit this parallelism to obtain a speedup in assembly time compared to the seeded model.

It is worthwhile to note that our main theorem does not apply to the seeded model.
For instance, it is routine to design a seeded tile system that assembles a unique terminal assembly shaped like a square, which uses the same tile type in the upper right and lower left corners of the square.
Translating this assembly to overlap those two positions means that this tile system satisfies the hypothesis of our main theorem.
Why does this not contradict the fact that this system, like all seeded systems, can be simulated by a hierarchical tile system at scale factor 5~\cite[Theorem 4.2]{TwoHandsBetterThanOne}, which would apparently satisfy the same consistent overlap condition?
The answer is that the hierarchical simulating system of~\cite{TwoHandsBetterThanOne} uses different $5 \times 5$ blocks to represent the same tile type from the seeded system, depending on the \emph{sides} of the tile that are used to bind in the seeded system.
Since the upper-right corner tile and lower-left corner tile in the seeded system must clearly bind using different sides, they are represented by different blocks in the simulating hierarchical system.
Hence in the hierarchical system, the terminal assembly does \emph{not} consistently overlap with itself.

Our argument proceeds by reducing the problem (via a simple argument) to a simple-to-state theorem in pure geometry.
That theorem's proof contains almost all of the technical machinery required to prove our main theorem.
Let $S_0$ be a discrete \emph{shape}: a finite, connected subset of $\Z^2$, and let $\vv \in \Z^2$ be a nonzero vector.
Let $S_1 = S_0 + \vv$ ($= \{\ p + \vv \ |\ p \in S_1 \ \}$), and let $S_2 = S_1 + \vv$.
The theorem states that $S_2 \setminus S_1$ (possibly a disconnected set) contains a connected component that does not intersect $S_0$.
This is clear when $\vv$ is large enough that $S_0 \cap S_2 = \emptyset$, but for the general case, we encourage the reader to attempt to prove it before concluding that it is obvious.
In Figure~\ref{fig:pump}, $S_2 \setminus S_1$ (referring respectively to the shapes of assemblies $\alpha_2$ and $\alpha_1$) contains two connected components, one on top and the other on bottom.
The top component intersects $S_0$, but not the bottom.

This problem is in turn reduced to a more technical statement about simple curves (continuous, one-to-one functions $\varphi:[0,1] \to \R^2$) whose intersection implies the shapes theorem. 
Although we need the curve theorem to hold only for polygonal curves on the integer grid $\Z^2$, the result holds for general curves, and we hope it may be useful in other contexts.


\section{Informal definition of the hierarchical tile assembly model}
\label{sec-tam-informal}

We give an informal sketch of the hierarchical variant of the abstract Tile Assembly Model (aTAM).
See Section \ref{sec-tam-formal} for a formal definition.


Let $\RR $, $\ZZ $, $\NN$ and $\ZZ^{+} $ denote the set of all real numbers, integers, non-negative integers and positive integers, respectively.
Given a set $S \subseteq \RR^{2} $ and a vector $\vv\in \RR^{2}$, let $S +
\vv = \{p + \vv :\: p\in S\}$.

A \emph{tile type} is a unit square with four sides, each consisting of a \emph{glue label} (often represented as a finite string). Each glue type is assigned a nonnegative integer \emph{strength}.
We assume a finite set $T$ of tile types, but an infinite number of copies of each tile type, each copy referred to as a \emph{tile}.
An \emph{assembly}
is a positioning of tiles on the integer lattice $\Z^2$; i.e., a partial function $\alpha:\Z^2 \dashrightarrow T$. 
Write $\alpha \sqsubseteq \beta$ to denote that $\alpha$ is a \emph{subassembly} of $\beta$, which means that $\dom\alpha \subseteq \dom\beta$ and $\alpha(p)=\beta(p)$ for all points $p\in\dom\alpha$. Given an assembly $\beta$ and a set $D\subseteq\dom\beta$, $\beta{\restriction_{D}}$ is a subassembly of $\alpha$ with $\dom(\beta{\restriction_{D}}) =  D$.

We abuse notation and take a tile type $t$ to be equivalent to the single-tile assembly containing only $t$ (at the origin if not otherwise specified).
Two adjacent tiles in an assembly \emph{interact} if the glue labels on their abutting sides are equal and have positive strength. 
Each assembly induces a \emph{binding graph}, a grid graph whose vertices are tiles, with an edge between two tiles if they interact.
The assembly is \emph{$\tau$-stable} if every cut of its binding graph has strength (the sum of the weights of the edges in the cut) at least $\tau$, where the weight of an edge is the strength of the glue it represents.

A \emph{hierarchical tile assembly system} (hierarchical TAS) is a pair $\calT = (T,\tau)$, where $T$ is a finite set of tile types and $\tau \in \N$ is the temperature.
An assembly is \emph{producible} if either it is a single tile from $T$, or it is the $\tau$-stable result of translating two producible assemblies without overlap.
The restriction on overlap is a model of a chemical phenomenon known as \emph{steric hindrance}~\cite[Section 5.11]{WadeOrganicChemistry91} or, particularly when employed as a design tool for intentional prevention of unwanted binding in synthesized molecules, \emph{steric protection}~\cite{HellerPugh1,HellerPugh2,GotEtAl00}.
An assembly $\alpha$ is \emph{terminal} if for every producible assembly $\beta$, $\alpha$ and $\beta$ cannot be $\tau$-stably attached.
If $\alpha$ can grow into $\beta$ by the attachment of zero or more assemblies, then we write $\alpha \to \beta$. Our definitions imply only finite assemblies are producible.
Figure~\ref{fig:2handed} shows an example of hierarchical attachment.

\begin{figure}[!ht]
  \vspace{-0.1in}
  \begin{center}
    \includegraphics[width=4in]{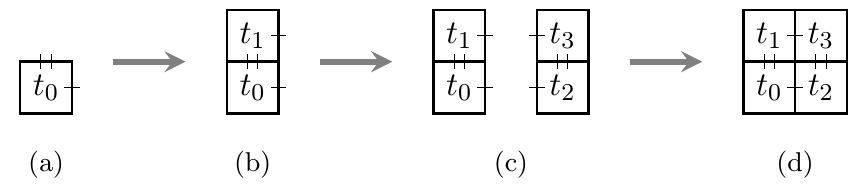}
  \caption{\figuresize Typical example of hierarchical assembly, at temperature $\tau = 2$. The segments between tiles represent the bonds, the number of segments encodes the strength of the bond (here,~$1$ or~$2$). In the seeded, single tile model with seed $\sigma = t_0$, the assembly at step (b) would be terminal.}
  \label{fig:2handed}
  \end{center}
  \vspace{-0.1in}
\end{figure}



\section{Main result}

\newcommand{\crv}[1]{\varphi_{#1}}

\newcommand{\crvt}[2]{\varphi_{#1}^{\rightarrow #2}}

\newcommand{\psit}[1]{\psi^{\rightarrow #1}}

\newcommand{\phit}[1]{\varphi^{\rightarrow #1}}

Section~\ref{subsec:curves} proves a theorem about curves in $\R^2$ (Theorem~\ref{t:k-lines}) that contains most of the technical detail required for our main theorem.
Section~\ref{subsec:shapes} uses Theorem~\ref{t:k-lines} to prove a geometrical theorem about shapes in $\Z^2$ (Theorem~\ref{thm:shapes}).
Section~\ref{subsec:assembly} uses Theorem~\ref{thm:shapes} to prove our main theorem (Theorem~\ref{thm:main}).

The high-level intuition of the proofs of these results is as follows (described in reverse order).
Theorem~\ref{thm:main} intuitively holds by the following argument.
If a producible assembly $\alpha_0$ is consistent with its translation $\alpha_1 = \alpha+\vv$ by some nonzero vector $\vv \in \Z^2$, then Theorem~\ref{thm:shapes} implies that some portion $C$ of $\alpha_2 = \alpha_0 + 2 \vv$ does not intersect $\alpha_0$, and $C$ is furthermore assemblable from $\alpha_1$ (by Theorem~\ref{thm-union-producible}).
Therefore, it is assemblable from $\alpha_0 \cup \alpha_1$ (since $\alpha_2$ is consistent with $\alpha_1$, and this part $C$ of $\alpha_2$ does not intersect $\alpha_0$, ruling out inconsistency due to $C$ clashing with $\alpha_0$).
Thus $\alpha_1 \cup C$ is producible and overlaps consistently with its translation by $\vv$.
Since $C$ is nonempty, $\alpha_1 \cup C$ is strictly larger than $\alpha_0$.
Iterating this argument shows that arbitrarily large producible assemblies exist.

Why does Theorem~\ref{thm:shapes} hold?
If it did not, then every connected component $C_i$ of $S_2 \setminus S_1$ would intersect $S_0$ at a point $p_i$.
Since $p_i \in S_0$, $p_i + 2\vv \in S_2$.
Since $p_i \in S_2$, there is a path $q_i$ from $p_i$ to $p_i + 2\vv$ lying entirely inside of $S_2$.
But Corollary~\ref{cor:one-line} implies that $q_i$ must intersect $q_i - \vv$, which, being a path inside of $S_1$, implies that $p_i + 2 \vv$ is in a different connected component $C_{i+1}$ of $S_2 \setminus S_1$.
But since $C_{i+1}$ also intersects $S_0$, there is a point $p_{i+1}$ in this intersection, and there is a curve $\varphi_{i}$ from $p_i + 2\vv$ to $p_{i+1}$.
Since every connected component of $S_2 \setminus S_1$ intersects $S_0$, we can repeat this argument until we return to the original connected component $C_i$.
But then the various curves $\varphi_i$ defined within each component will satisfy the conditions of Theorem~\ref{t:k-lines}, a contradiction.

\subsection{A theorem about curves}
\label{subsec:curves}

\begin{definition}
Given a nonzero vector $\vec v\in\RR^2$ and a point $p\in\RR^2$ the $\vec v$-\emph{axis through} $p$,  denoted as $L_{\vec v,p}$, is the line parallel to $\vec v$ through $p$.
\end{definition}

\begin{definition}
Let $\varphi:[0,1]\to\RR^2$ be continuous one-to-one mapping. Then $\varphi([0,1])$ is called a \emph{simple} (non-self-intersecting) curve from $\varphi(0)$ to $\varphi(1)$. If $\varphi:[0,1]\to\RR^2$ is continuous with $\varphi(0)=\varphi(1)$ and one-to-one on $[0,1)$, then $\varphi([0,1])$ is called a \emph{simple closed} curve.
\end{definition}

Obviously, any curve $\varphi([0,1])$ from $\varphi(0)$ to $\varphi(1)$ (being a subset of the plane) can be considered also as a curve from $\varphi(1)$ to $\varphi(0)$. Therefore, for the sake of brevity, we sometimes denote this curve simply by $\varphi$ and say that $\varphi$ connects points $\varphi(0)$, $\varphi(1)$. If $0 \leq t_1 \leq t_2 \leq 1$, then $\varphi([t_1,t_2])$ is a simple curve as well. If $\varphi_1$ and $\varphi_2$ are simple non-closed curves such that $\varphi_1 \cap \varphi_2 = \{ \varphi_1(1)\} = \{\varphi_2(0) \}$ then their concatenation $\varphi_1 \oplus \varphi_2$, defined by $(\varphi_1 \oplus \varphi_2)(t) = \varphi_1(2t)$ if $t \leq \frac{1}{2}$ and $(\varphi_1 \oplus \varphi_2)(t) = \varphi_2(2(t - \frac{1}{2}))$ otherwise, is also a simple curve (closed if $\varphi_2(1) = \varphi_1(0)$).

\begin{definition}
Given a subset of a plane $A\subseteq\RR^2$ and a vector $\vec v\in\RR^2$, the \emph{shift (or translation) of $A$ by $\vec v$}, denoted by $A+\vec v$, is the set
$
A+\vec v=\{p+\vec v : p\in A\}.
$
\end{definition}


The following lemma states that if a curve does not intersect a translation of itself, then it also does not intersect any integer multiples of the same translation.
A similar lemma was proven independently by Demaine, Demaine, Fekete, Patitz, Schweller, Winslow, and Woods~\cite{OneTile} for shapes instead of curves.

\begin{lemma}\label{l:one-stripe}
Consider points $p_1,p_2\in\RR^2$, nonzero vector $\vec v\in\RR^2$ and a simple curve $\varphi$ connecting $p_1$ and $p_2$ ($\varphi$ may be closed if $p_1=p_2$) such that $\varphi\cap(\varphi+\vec v)=\emptyset$. Let $\phit{k}=\varphi + k\vec v,\ k\in\ZZ$. Then all $\phit{k}$'s are mutually disjoint.
\end{lemma}

\begin{proof}
To every point of $\varphi$ we can assign ``relative distance'' $d$ from the line $L_{\vec v,p_1}$---positive for points left to the line and negative for points right to the line (with respect to $\vec v$). Since the function $d \circ \varphi : [0,1]\to\RR$ is continuous, by the extreme value theorem it attains both its minimum $d_\mathrm{min}$ and maximum $d_\mathrm{max}$.

If $d_\mathrm{min}=d_\mathrm{max}$ then $\varphi$ is just a line segment on the line $L_{\vec v,p_1}$ with a length less than $|\vec v|$ and the statement of the lemma holds true.

If $d_\mathrm{min}<d_\mathrm{max}$, let
$T_\mathrm{min} =\{t\in[0,1] : d\circ \varphi(t) = d_\mathrm{min}\}$ and
$T_\mathrm{max} =\{t\in[0,1] : d\circ \varphi(t) = d_\mathrm{max}\}$.
Since both $T_\mathrm{min}$ and $T_\mathrm{max}$ are closed and non-empty, we can take $t_\mathrm{min}\in T_\mathrm{min}$ and $t_\mathrm{max}\in T_\mathrm{max}$ such that
$d_\mathrm{min} < d\circ\varphi(t) < d_\mathrm{max}$ for every
$t\in(\min\{t_\mathrm{min}, t_\mathrm{max}\}, \max\{t_\mathrm{min}, t_\mathrm{max}\})$. Denote $p_\mathrm{min}=\varphi(t_\mathrm{min})$ and $p_\mathrm{max}=\varphi(t_\mathrm{max})$. All curves $\phit{k}$, $k\in\ZZ$, lie within the stripe between lines $L_{\vec v,p_\mathrm{min}}$ and $L_{\vec v,p_\mathrm{max}}$. Denote
$\psi=\varphi([\min\{t_\mathrm{min}, t_\mathrm{max}\}, \max\{t_\mathrm{min}, t_\mathrm{max}\}])$ (a simple curve connecting $p_\mathrm{min}$ and $p_\mathrm{max}$) and let $\psit{k}=\psi+k\vec v$, $k\in\ZZ$, be the corresponding shifts of $\psi$.

Since $\psit{k}$ meets neither $L_{\vec v,p_\mathrm{min}}$ nor $L_{\vec v,p_\mathrm{max}}$ at any point except its end-points, it splits the stripe into two disjoint regions---left and right (with respect to vector $\vec v$)---let us denote the left region by $L_k$ and the right one by $R_k$.

Since $\varphi \cap (\varphi + \vec v) = \emptyset $, we have for
every $k\in \ZZ $, $\psit{k}\cap \phit{k + 1}\subseteq \phit{k}\cap \phit{k + 1} = \emptyset $. Since the point $p_\mathrm{min}+(k+1)\vec v\in\phit{k+1}$ lies in $R_k$ and $\phit{k+1}\cap\psit{k}=\emptyset$, the whole curve $\varphi_{k+1}$ lies in $R_k$. Hence $\psit{k+1}\subseteq R_k$ and similarly $\psit{k-1}\subseteq\phit{k-1}\subseteq L_k$. This yields $R_{k+1}\subseteq R_k$ and $L_{k-1}\subseteq L_k$ and consequently $R_\ell\subseteq R_k$ and $L_k\subseteq L_\ell$ for any $k\le \ell$, $k,\ell\in\ZZ$.

Consider now any $k<\ell$, $k,\ell\in\ZZ$. If $\ell=k+1$ then $\phit{k} \cap \phit{\ell} = \emptyset$ by the  assumption of the lemma. If $\ell>k+1$ then $\phit{k} \subseteq L_{k+1}$ and $\phit{\ell} \subseteq R_{\ell-1} \subseteq R_{k+1}$, i.e., $\phit{k}$ and $\phit{\ell}$ are disjoint.
\end{proof}

The following corollary of Lemma~\ref{l:one-stripe} shows that if a curve is translated by a vector $\vv$, and the vector between its start and end points is an integer multiple of $\vv$, then the curve must intersect its translation by $\vv$.

\begin{cor}\label{cor:one-line}
Consider an integer $n\geq 1$, a point $p\in\RR^2$ and a nonzero vector $\vec v\in\RR^2$. Let $\varphi$ be a simple curve connecting $p$ and $p + n\vec v$. Then $\varphi$ intersects its translation by $\vec v$.
\end{cor}
\begin{proof}
Assume for the sake of contradiction that $\varphi$ and $\varphi +
\vec v$ do not intersect. By Lemma~\ref{l:one-stripe} all curves $\varphi+n\vec v$, $n\in\NN$, are mutually disjoint but $(p+n\vec v)\in\varphi \cap (\varphi + n \vv)$---a contradiction.
\end{proof}

The assumption that the vector from the start point to the end point of the curve $\varphi$ is an integer multiple of the vector $\vec v$ is essential in Corollary~\ref{cor:one-line}. The following example provides a general construction of a curve $\varphi\subseteq\RR^2$ connecting points $p$ and $p+x\vec v$ such that
$\varphi\cap(\varphi+\vec v)=\emptyset$, where $\vec 0\neq\vec v\in\RR^2$ and $x\in\RR\setminus\ZZ$, $|x|>1$. Note that for $|x|<1$ the line segment from $p$ to $p+x\vec v$ does not intersect its shift by $\vec v$.

\begin{example}
For simplicity assume that $p=(0,0)$ and $\vec v=(1,0)$. Let $n = \lfloor x \rfloor$, $y = x - n$ and choose any $\varepsilon>0$.

Let $\mu$ denote the line segment (simple curve) from $(0,0)$ to $(y,
n\varepsilon)$ and $\nu$ denote the line segment from $(y,
n\varepsilon)$ to $(1, -\varepsilon)$. Denote $\mu_k=\mu+k(1,
-\varepsilon)$ and $\nu_k=\nu+k(1, -\varepsilon)$ for $k\in\ZZ$. Then
let $\varphi=\mu_0 \oplus \nu_0 \oplus \dots \oplus \mu_{n-1} \oplus
\nu_{n-1} \oplus \mu_n$ be the desired curve. Figure~\ref{f:non-integer}
shows an example of this construction for $x = 3.6$. Note that $\varphi $ starts
and ends on the $x$-axis and that $\varphi + \vv $ does not intersect
$\varphi $ since for each stripe between $x = i$ and $x = i + 1$, $i = 1,\dots,n$, the
part of $\varphi + \vv $ in this stripe lies above the part of $\varphi
$ in the same stripe (shifted up by $\varepsilon $).

\begin{figure}[h!]
\centering
\includegraphics[width=\textwidth]{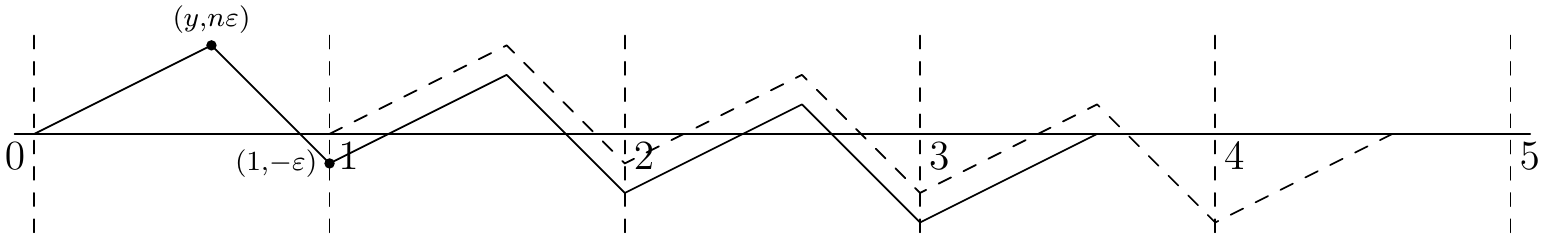}
\caption{An example of a curve $\varphi$ from $(0,0)$ to $(3.6,0)$
  (solid) that does not intersect its shift $\varphi +  (1,0)$ (dashed)}
\label{f:non-integer}
\end{figure}
\end{example}


Note that by Lemma \ref{l:one-stripe} all shifts $\crvt{}{\ell} = \varphi+\ell\vec v,\ \ell\in\ZZ$ are mutually disjoint.

The following theorem is quite technical to state.
Informally, it concerns a finite set of non-intersecting curves $\crv{1},\ldots,\crv{k}$ and a vector $\vv$ of the following form.
Each curve connects two points in the plane, subject to the condition that the end point of $\crv{i}$ is the start point of $\crv{i+1}$ translated by a positive integer multiple of $\vv$, with $\crv{k+1} = \crv{1}$.
See Figure~\ref{f:example-multiple-curves}(a) for an example.
An alternative way to think of these curves is as a single ``mostly continuous'' simple closed curve, with $k$ discontinuities allowed, where each discontinuity is of the form ``jump backwards by some positive integer multiple of $\vv$.''
The theorem states that this curve must intersect its translation by $\vv$.

\begin{theorem}\label{t:k-lines}
Let $k \in \Z^+$, let $p_1,\dots,p_{k} \in \RR^2$ be points, let $n_1,\dots,n_k \in \Z^+$, and let $\vec v\in\RR^2$ be a nonzero vector.
Then there do not exist curves $\crv{1},\dots,\crv{k}:[0,1] \to \R^2$ satisfying the following conditions:
\begin{enumerate}
\item \label{cond-start-end-points} $\crv{i}$ is a simple curve from
  $p_i$ to $(p_{i+1}+n_{i+1}\vec v)$, for every $1\leq i\leq k$, where $p_{k+1}=p_1$ and $n_{k+1}=n_1$,
\item \label{cond-no-intersect-translation} $\crv{i} \cap
  (\crv{i}+\vec v)=\emptyset$, for every $1\leq i\leq k$,
\item \label{cond-no-intersect-other-curves} $\crv{i} \cap \crv{j} =
  \emptyset$, for every $1\leq i < j\leq k$.
\end{enumerate}
\end{theorem}

\begin{proof}
By induction on $k$. The base case $k=1$ immediately follows by Corollary~\ref{cor:one-line}.

Intuitively, the inductive case will show that if we suppose, for the sake of contradiction, that $k$ curves exist satisfying the conditions, then we can find a common point of intersection between two of their integer translations by $\vv$, and we can connect two subcurves of these translations to create a set of $k' < k$ curves also satisfying the hypothesis of the theorem, without introducing an intersection.
Figure~\ref{f:example-multiple-curves} shows an example of three curves being reduced to two.
The new curves will simply be $k' - 1$ translations of some of the original $k$ curves (which already satisfy the conditions by hypothesis), together with one new curve $\psi$, so our main task will be to show that $\psi$, in the presence of the other pre-existing curves, satisfies the three conditions.

More formally, let $k>1$ and suppose the theorem holds for all integers $0 < k' < k$. Assume for the sake of contradiction that there are curves $\crv{1},\ldots,\crv{k}$ satisfying conditions~\ref{cond-start-end-points},~\ref{cond-no-intersect-translation}, and~\ref{cond-no-intersect-other-curves}, and define $\crvt{m}{\ell} = \crv{m} + \ell \vv$ for all $m \in\{1,\ldots,k\}$ and $\ell \in\NN$. We find the first intersection of $\crv{1}$ with any of curves $\crvt{m}{\ell}$ for all $m \in \{2,\ldots,k\}$ and $\ell\in\NN$. Let
\begin{align*}
t_1 &= \min\{t\in[0,1] : (\exists m \in \{2,\ldots,k\})(\exists \ell\in\NN)\,\crv{1}(t) \in \crvt{m}{\ell} \}, \\
M &= \text{any } m \in \{2,\ldots,k\} \text{ such that } (\exists \ell\in\NN)\,\crv{1}(t_1) \in \crvt{m}{\ell}, \\
L &= \text{the unique } \ell\in\NN \text{ such that } \crv{1}(t_1) \in \crvt{M}{\ell}, \\
t_2 &= \text{the unique } t\in[0,1] \text{ such that } \crv{1}(t_1) = \crvt{M}{L}(t).
\end{align*}

\begin{SCfigure}[50]
\includegraphics[width=0.5\textwidth]{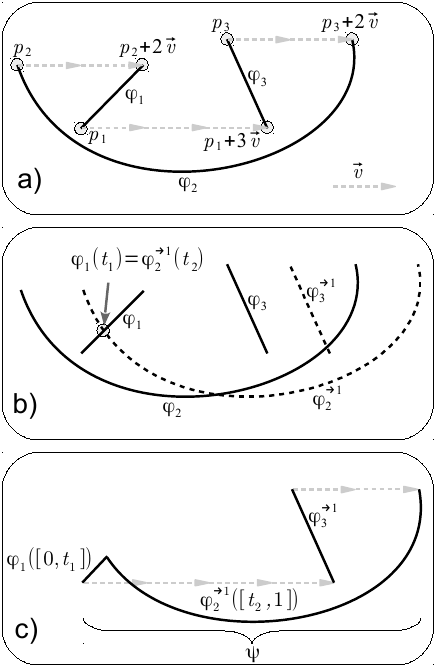}
\caption{\figuresize An example of the proof of Theorem~\ref{t:k-lines} for $k=3$ curves.
\\(a) Three curves, $\crv{1}$, $\crv{2}$, and $\crv{3}$, with start
and end points obeying condition~\ref{cond-start-end-points} and also condition~\ref{cond-no-intersect-other-curves} (the curves violate condition~\ref{cond-no-intersect-translation}, however, as Theorem~\ref{t:k-lines} dictates they must if obeying the other two conditions).
In this case, $n_1 = 3$, $n_2 = 2$, and $n_3 = 2$.
\\(b) Translations of curves $\crv{2}$ and $\crv{3}$ by $\vv$, showing that $\crv{1}$ first intersects $\crvt{2}{1}$, among all positive integer translations of $\crv{2}$ and $\crv{3}$. So in this example, $M=2$ and $L=1$.
\\(c) $\psi$ defined as the concatenation of $\crv{1}([0,t_1])$ with $\crvt{2}{1}([t_2,1])$. $\psi$ and $\crvt{3}{1}$ and are the two curves produced by the proof for the inductive argument.
}
\label{f:example-multiple-curves}
\end{SCfigure}

Since $\crv{1}$ intersects $\crvt{2}{n_2}$ at $p_{2} + n_{2}\vv $ by condition~\ref{cond-start-end-points}, $t_1$, $M$, and $L$ are well-defined.
The uniqueness of $L$ follows by Lemma~\ref{l:one-stripe}.
The uniqueness of $t_2$ follows from the fact that $\crvt{M}{L}$ is simple.

Now define the curve $\psi$ as a concatenation
\begin{align*}
\psi &= \crv{1}([0,t_1]) \oplus \crvt{M}{L}([t_2,1])\\
\noalign{\noindent and consider its shift}
\psi + \vec v &= \crvt{1}{1}([0,t_1]) \oplus \crvt{M}{L+1}([t_2,1]).
\end{align*}

In what follows we will show that points $p_1,p_{M+1}+L\vv,\dots,p_k+L\vv$, integers
$n_1+L,n_{M+1},\dots,n_k$ and curves
$\psi,\crvt{M+1}{L},\dots,\crvt{k}{L}$ form another instance
satisfying conditions~\ref{cond-start-end-points},~\ref{cond-no-intersect-translation}, and~\ref{cond-no-intersect-other-curves}.

Observe that $\psi$ is a curve connecting the point $p_1$ to the point $p_{M+1}+(n_{M+1}+L)\vv$.
It consists of subcurves of two simple curves whose concatenation at the intersection point $\crv{1}(t_1) = \crvt{M}{L}(t_2)$, by the definition of $t_1$, is the first point of intersection between $\crv{1}$ and $\crvt{M}{L}$.
The curve $\crvt{M}{L}$ after that point (i.e., $\crvt{M}{L}((t_2,1])$) therefore cannot intersect $\crv{1}([0,t_1))$, so $\psi$ is simple.
It follows that $\psi$ satisfies condition~\ref{cond-start-end-points}
of the new instance.

We establish that $\psi$ does not intersect its shift by vector $\vv$ by analyzing each of the two parts of $\psi$, $\crv{1}([0,t_1])$ and $\crvt{M}{L}([t_2,1])$, and their translations by $\vv$, separately:
\begin{itemize}
\item $\crv{1}([0,t_1)) \cap \crvt{1}{1}([0,t_1)) = \emptyset$, since $\crv{1} \cap \crvt{1}{1} = \emptyset$ by condition~\ref{cond-no-intersect-translation}.
\item $\crvt{M}{L}([t_2,1]) \cap \crvt{M}{L+1}([t_2,1])=\emptyset$, since
it follows by condition~\ref{cond-no-intersect-translation} that $\crvt{M}{L}\cap \crvt{M}{L + 1} = \emptyset$.
\item $\crv{1}([0,t_1)) \cap \crvt{M}{L+1}([t_2,1]) = \emptyset$,
  since by the definition of $t_1$ (in particular, the fact that it is
  the minimum element of the set defining it), $\crv{1}([0,t_1))$ does
  not intersect any  $\crvt{m}{\ell}$, for any $m \geq 2, \ell\in\NN$.
\item $\crvt{M}{L}([t_2,1]) \cap \crvt{1}{1}([0,t_1))=\emptyset$, since otherwise $\crv{1}([0,t_1))$ would intersect $\crvt{M}{L-1}$, violating the definition of $t_1$ similarly to the previous point.
\end{itemize}
This implies that $\psi$ satisfies condition~\ref{cond-no-intersect-translation}.

\begin{figure}[h!]
\centering
\includegraphics{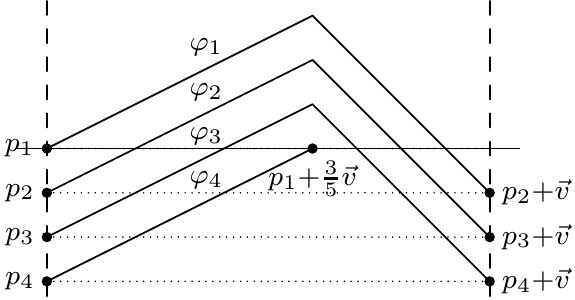}
\caption{\figuresize
An example of four curves $\varphi_{1},\dots,\varphi_{4}$ that satisfy the conditions of Theorem~\ref{t:k-lines}, except that $n_{1} = \frac{3}{5}$ is not an integer.}
\label{f:non-integer-k}
\end{figure}

We have $\crvt{i}{L} \cap \psi = \emptyset$ for every $i>M$, since
$\crvt{i}{L}$ cannot intersect $\crv{1}([0,t_1))$ (by definition of
$t_{1}$) and $\crvt{i}{L} \cap \crvt{M}{L} = \emptyset$ by condition~\ref{cond-no-intersect-other-curves}.
This implies that $\psi$ satisfies
condition~\ref{cond-no-intersect-other-curves} of the new instance.

Thus, the new instance with points $p_1,p_{M+1}+L\vv,\dots,p_k+L\vv$,
integers $n_1+L,n_{M+1},\dots,n_k$ and curves
$\psi,\crvt{M+1}{L},\dots,\crvt{k}{L}$ satisfy
conditions~\ref{cond-start-end-points},~\ref{cond-no-intersect-translation},
and~\ref{cond-no-intersect-other-curves}. In addition, it has a
smaller number of curves ($k+1-M=k'<k$), and hence, using the induction hypothesis we have a contradiction.
\end{proof}


The example in Figure~\ref{f:non-integer-k} shows that the theorem does not hold if we allow just one of the numbers $n_{1},\dots,n_{k}$ to be a non-integer.

%

\subsection{A theorem about shapes}
\label{subsec:shapes}

\begin{figure}[h!]
\centering
\includegraphics[width=\textwidth]{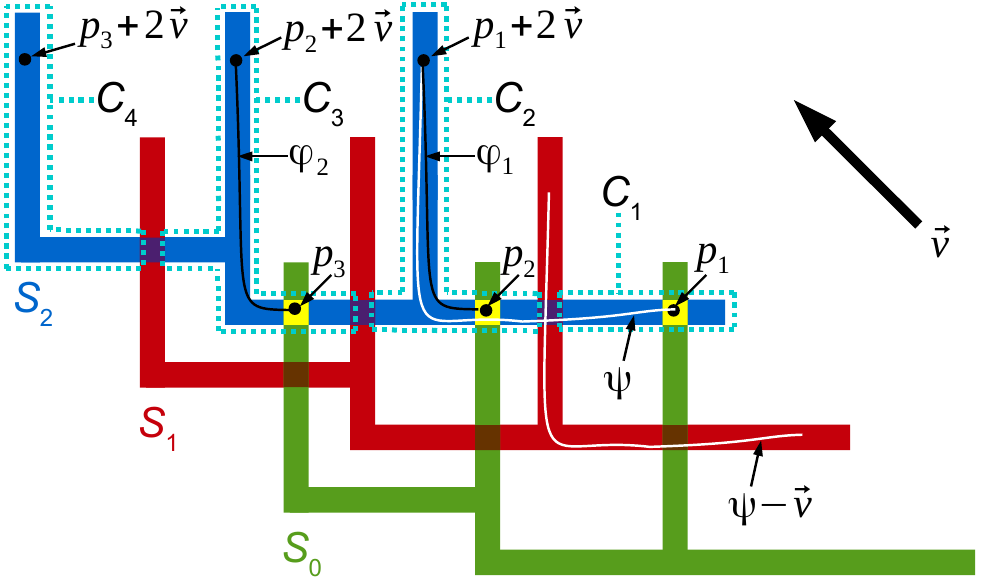}
\caption{\figuresize An example of a shape $S_0$ and its two translations. Starting at $p_1 \in (S_2 \cap S_0) \setminus S_1$, we repeat the following procedure: from point $p_i$ in connected component $C_i$ of $S_2 \setminus S_1$, jump to point $p_i + 2\vv$, which is guaranteed to be in a different connected component $C_{i+1}$ of $S_2 \setminus S_1$ from $p_i$ (see proof of Theorem~\ref{thm:shapes} to see why this is implied by Corollary~\ref{cor:one-line}). If $C_{i+1}$ intersects $S_0$ at point $p_{i+1}$, then there is a curve $\varphi_i$ in $S_2 \setminus S_1$ from $p_i + 2\vv$ to $p_{i+1}$, and jumping to point $p_{i+1} + 2\vv$ takes us to yet another connected component $C_{i+2} \neq C_{i+1}$. Repeating this must eventually result in a connected component (in this example, $C_4$) that does not intersect $S_0$, or else the curves $\varphi_i$ would contradict Theorem~\ref{t:k-lines}.}
\label{fig:shape-example}
\end{figure}

Theorem~\ref{t:k-lines} gives rise to the following geometrical theorem about discrete shapes, which is the main technical tool to prove our main self-assembly result, Theorem~\ref{thm:main}.
We define a \emph{shape} to be a finite, connected subset of $\Z^2$.

\begin{theorem}\label{thm:shapes}
  Let $S_0 \subset \Z^2$ be a shape, and let $\vv \in \Z^2$ be a nonzero vector.
  Let $S_1 = S_0 + \vv$ and $S_2 = S_1 + \vv$.
  Then there is a connected component of $S_2 \setminus S_1$ that does not intersect $S_0$.
\end{theorem}

\begin{proof}

  We first sketch an informal intuition of the proof, shown by example in Figure~\ref{fig:shape-example}.
  The argument is constructive: it shows a way to iterate through some connected components of $S_2 \setminus S_1$ to actually find one that does not intersect $S_0$.
  
  Start with component $C_1$, and suppose it intersects $S_0$ at point $p_1 \in C_1 \cap S_0$.
  Then $p_1 + 2\vv \in S_2$ since $p_1 \in S_0$.\footnote{In this example $p_1 + 2\vv \not \in S_1$; in the full argument we consider $p_1 + n \vv$ for $n \in \Z^+$ large enough to ensure this.}
  Let $\psi$ be a path (simple curve) from $p_1$ to $p_1 + 2\vv$ lying entirely within $S_2$.
  Corollary~\ref{cor:one-line} implies that $\psi$ intersects $\psi - \vv$, which is a curve lying entirely within $S_1$.
  In other words, every path from $p_1$ to $p_1 + 2\vv$ lying inside $S_2$ hits $S_1$, i.e., $p_1 + 2\vv$ and $p_1$ are in different connected components of $S_2 \setminus S_1$.
  We call $C_2 \neq C_1$ the connected component of $p_1 + 2\vv$.
  Suppose $C_2$ also intersects $S_0$; then there is some curve $\varphi_1$ lying entirely within $S_2 \setminus S_1$ and going from $p_1 + 2\vv$ to this new point $p_2 \in C_2 \cap S_0$.
  Repeating the previous argument, $p_2 + 2 \vv$ must be in a different connected component $C_3 \neq C_2$, and if $C_3$ also intersects $S_0$, then there is another curve $\varphi_2 \subset C_3$ from $p_2 + 2\vv$ to $p_3 \in C_3 \cap S_0$.
  In this example, we iterate this one more time and find that connected component $C_4 \subset S_2 \setminus S_1$ does not intersect $S_0$.
  
  For the sake of contradiction, suppose that we fail to find such a connected component, i.e., every one of the connected components $C_1,\ldots,C_k$ of $S_2 \setminus S_1$ intersects $S_0$.
  Then eventually the above described procedure cycles back to a previously visited connected component, and the curves $\varphi_j$ contained in $S_2 \setminus S_1$ satisfy condition~\ref{cond-start-end-points} of Theorem~\ref{t:k-lines}.
  Since each $\varphi_i \in S_2 \setminus S_1$, we have $\varphi_i + \vv \in S_3 \setminus S_2$, hence $\varphi_i \cap (\varphi_i + \vv) = \emptyset$ for all $1 \leq i \leq k$, so they satisfy condition~\ref{cond-no-intersect-translation}.
  Since each curve lies in a different connected component of $S_2 \setminus S_1$, they do not intersect each other, satisfying condition~\ref{cond-no-intersect-other-curves}, a contradiction.
  
  More formally, consider connected components of $S_{2}\setminus S_{1}$, say $C_{1},\dots,C_{k}$, for some $k\ge 1$.
  We say that $C_{i}$ is \emph{non-conflicting} if $C_{i}\cap S_{0} = \emptyset $. We will
  show that there is a non-conflicting $C_{i}$.
  Assume for the sake of contradiction that for every $i = 1,\dots,k$, $C_{i}\cap S_{0}\ne \emptyset $ and let $p_{i}\in C_{i}\cap S_{0}$.
  Note that $p_{i} + \vec v\in  S_{1}$.
  Let $n_{i}$ be the smallest positive integer such that $p_{i} + n_{i}\vec v\notin S_{1}$ (since $S_{1}$ is finite, such an $n_{i}$ must exist).
  Since $p_{i} + (n_{i} - 1)\vec v\in S_{1}$, we have $p_{i} + n_{i}\vec v\in S_{2}\setminus S_{1}$.
  Hence, $p_{i} + n_{i}\vec v$ belongs to some connected component of $S_2 \setminus S_1$. Both $p_i$ and $p_i + n_i \vv$ are in $S_2$, but by Corollary~\ref{cor:one-line}, any path within $S_2$ connecting them must intersect its translation by $-\vv$, which is a path in $S_1$, so $p_i + n_i \vv$ must be in a different connected component than $C_i$.
  We call this connected component $C_{i+1}$.\footnote{Assuming we do this for every point $p_i$, at some point we must cycle back to a connected component already visited. It may not be that this cycle contains all connected components of $S_2 \setminus S_1$, but in this case we consider $C_1,\ldots,C_k$ to be not every connected component of $S_2 \setminus S_1$, but merely those encountered in the cycle, so that for the sake of notational convenience we can assume that $C_1,\ldots,C_k$ are all encountered, and indexed by the order in which they are encountered.}

  Consider a simple curve (a self-avoiding path in the lattice) $\varphi_{i}$ from $p_{i}$ to $p_{i + 1} + n_{i + 1}\vec v$ in $C_{i}\subseteq S_{2}\setminus S_{1}$.
  Since these paths lie in different connected components they do not intersect. Furthermore, since $\varphi_{i} + \vec v \subset S_{3}\setminus S_{2}$, it does not intersect $\varphi_{i} \subset S_2$. But these curves contradict Theorem~\ref{t:k-lines}.
\end{proof}

\subsection{Implication for self-assembly}
\label{subsec:assembly}

In this section we use Theorem~\ref{thm:shapes} to prove our main theorem, Theorem~\ref{thm:main}.
We require the following theorem from~\cite{DotyProducibility}. 
We say that two overlapping assemblies $\alpha$ and $\beta$ are \emph{consistent} if $\alpha(p)=\beta(p)$ for every $p\in\dom\alpha\cap\dom\beta$.
If $\alpha$ and $\beta$ are consistent, define their \emph{union} $\alpha \cup \beta$ to be the assembly with $\dom (\alpha \cup \beta) = \dom \alpha \cup \dom \beta$ defined by $(\alpha \cup \beta)(p) = \alpha(p)$ if $p \in \dom \alpha$ and $(\alpha \cup \beta)(p) = \beta(p)$ if $p \in \dom \beta$.
Let $\alpha \cup \beta$ be undefined if $\alpha$ and $\beta$ are not consistent.

\begin{theorem}[\cite{DotyProducibility}]\label{thm-union-producible}
  If $\alpha $ and $\beta $ are $\mathcal{T}$-producible assemblies
  that are consistent and overlapping, then $\alpha \cup \beta$
  is $\mathcal{T}$-producible. Furthermore, it is possible to assemble
  first $\alpha $ and then assemble the missing portions of $\beta $,
  i.e., $\beta {\restriction_{C_{1}}},\dots,\beta
  {\restriction_{C_{k}}}$, where $C_{1},\dots,C_{k}$ are connected
  components of 
  $\dom \beta \setminus \dom \alpha.$
\end{theorem}

\begin{definition}
Let $\alpha+\vv$ denote the translation of $\alpha$ by $\vv$, i.e., an assembly $\beta$ such that $\dom\beta=\dom\alpha+\vv$ and $\beta(p)=\alpha(p-\vv)$ for all $p\in\dom\beta$.
  We say that assembly $\alpha$ is \emph{repetitious} if there exists a nonzero
  vector $\vec v\in \ZZ ^{2}$ such that $\dom \alpha \cap \dom (\alpha  + \vec
  v)\ne \emptyset $ and $\alpha $ and $\alpha  + \vec v$ are consistent.
\end{definition}

Note that Theorem~\ref{thm-union-producible} implies that if a producible assembly $\alpha$ is repetitious with translation vector $\vv$, then $\alpha \cup (\alpha + \vv)$ is also producible.
The following is the main theorem of this paper.

\begin{theorem}\label{thm:main}
  Let $\calT$ be a hierarchical tile assembly system.
  If $\calT$ has a producible repetitious assembly, then arbitrarily large assemblies are producible in $\calT$.
\end{theorem}

\begin{proof}
  It suffices to show that the existence of a producible repetitious assembly $\alpha$ implies the existence of a strictly larger producible repetitious assembly  $\alpha' \sqsupset \alpha$. 
  Let $\alpha$ be a producible repetitious assembly, with $\vv \in \Z^2$ a nonzero vector such that $\alpha$ and $\alpha + \vv$ overlap and are consistent.
  For all $i\in\{0,1,2\}$, let $\alpha _{i} = \alpha + i\vec v$ and $S_{i} = \dom \alpha _{i}$.
  
%
  By Theorem~\ref{thm:shapes}, at least one connected component $C_2 \subseteq S_2 \setminus S_1$ does not intersect $S_0$.
  Define $C_1 = C_2 - \vv$.
  Note that $C_1 \subseteq S_{1}\setminus S_{0}$, which implies, since $C_2 \subseteq S_2 \setminus S_1$, that $C_2 \cap C_1 = \emptyset$. Let $\bar\alpha = \alpha _1
  {\restriction_{C_1}}$. Define $\alpha ' = \alpha \cup
  \bar\alpha $.
  By Theorem~\ref{thm-union-producible}, $\alpha '$ is producible.
  Consider $\dom \alpha '\cap \dom (\alpha ' + \vec v)$; it suffices to show that $\alpha'$ and $\alpha'+\vv$ are consistent on every tile type in this intersection. We have
  \begin{eqnarray*}
    \dom \alpha '\cap \dom (\alpha ' + \vec v)
    &=&
    (S_0 \cup C_1) \cap (S_1  \cup C_2)
    \\&=&
    (S_0 \cap S_1) \cup (S_0 \cap C_2) \cup (C_1 \cap S_1) \cup (C_1 \cap C_2)
    \\&=&
    (S_0 \cap S_1) \cup \emptyset \cup (C_1 \cap S_1) \cup \emptyset
    \\&=&
    (S_0 \cap S_1) \cup C_1,
  \end{eqnarray*}
  By the hypothesis that $\alpha$ is repetitious, $\alpha'$ and $\alpha'+\vv$ are consistent on $S_0 \cap S_1$.
  Since $\bar\alpha \sqsubseteq \alpha+\vv$ and $\dom \bar \alpha = C_1$, this implies that $\alpha'$ and $\alpha'+\vv$ are consistent on $C_1$ as well.
  Hence $\alpha'$ is repetitious.
  Since $C_1 \subseteq S_{1}\setminus S_{0}$ and is nonempty, $|\dom
  \alpha '| > |\dom \alpha |$.
\end{proof}

\paragraph{Acknowledgements.}
The authors are extremely grateful to Jozef Hale\v{s} for the proof of Theorem~\ref{t:k-lines}.
Although Jozef requested not to be a coauthor, that theorem is the keystone of the paper.
The second author is also grateful to David Kirkpatrick, Pierre-\'{E}tienne Meunier, Damien Woods, Shinnosuke Seki, and Andrew Winslow for several insightful discussions. The third author would like to thank Sheung-Hung Poon for useful discussions.

\newpage
\bibliographystyle{plain}
\bibliography{tam}
\newpage

\newpage
\appendix
\section{Appendix}

\subsection{Formal definition of the hierarchical tile assembly model}
\label{sec-tam-formal}

We will consider the square lattice, i.e., the graph $L_{\square}$ with the vertex set $\mathbb{Z}^2$ and the edge set $\{ (u,v): |u,v| = 1\}$.
The directions $\mathcal{D} = \{N, E, S, W\}$ are used to indicate the natural directions in the lattice.
Formally, they are functions from $\mathbb{Z} \times \mathbb{Z}$ to $\mathbb{Z} \times \mathbb{Z}$: $N(x,y) = (x,y+1),$ $E(x,y) = (x+1,y),$ $S(x,y) = (x,y-1),$ and $W(x,y) = (x-1,y)$.
Note that $-E = W$ and $-N = S$.

Informally, a tile is a square with the north, east, south, and west edges labeled from some finite alphabet $\Sigma$ of \emph{glues}.
Formally, a tile $t$ is a 4-tuple $(g_N, g_E, g_S, g_W) \in \Sigma^4$, indicating the glues on the north, east, south, and west side, respectively.
%
%
Each pair of glues $g$ and $g'$ is associated with a nonnegative integer $str(g,g')$ called the {\em interaction strength}. 

An \textit{assembly} on a set of tiles $T$ is a partial map $\alpha: \mathbb{Z}^2 \dashrightarrow T$ such that the subgraph of $L_{\square}$ induced by the domain of $\alpha$, denoted by $L_{\square}[\dom \alpha]$, is connected.
The \emph{weighted subgraph induced by $\alpha$}, denoted by $L_{\square}[\alpha]$, is $L_{\square}[\dom \alpha]$ in which every edge $pq$ has weight equal to the interaction strength of the glues on the abutting sides of tiles at positions $p$ and $q$, respectively, i.e., $str(\alpha(p)_d, \alpha(q)_{-d})$ where $d = q-p$.
Given a positive integer $\tau \in \mathbb{Z}^+$, called a \emph{temperature}, a set of edges of $L_{\square}[\alpha ]$ of an assembly $\alpha $ is \emph{$\tau$-stable} if the sum of the weights of edges in this set is at least $\tau$, and assembly $\alpha$ is \emph{$\tau$-stable}
if every edge cut of $L_{\square}[\alpha]$ is $\tau$-stable.

A {\em hierarchical tile assembly system} (hierarchical TAS) is a triple $\mathcal{T} = (T,\tau ,str)$, where $T$ is a finite set of tile types, $\tau \in \mathbb{Z}^{+}$ and $str:\: \Sigma \times \Sigma \to \NN $ is the interaction strength function.
Let $\alpha ,\beta :\: \ZZ^{2} \to T$ be two assemblies. 
We say that $\alpha $ and $\beta $ are \emph{nonoverlapping} if $\dom \alpha \cap \dom \beta = \emptyset$.
Two assemblies $\alpha $ and $\beta $ are \emph{consistent} if $\alpha(p) = \beta (p)$ for all $p\in \dom \alpha \cap \dom \beta $.
If $\alpha $ and $\beta $ are consistent assemblies, define the assembly $\alpha \cup \beta $ in a natural way, i.e., $\dom(\alpha \cup \beta ) = \dom \alpha \cup \dom \beta $ and $(\alpha \cup \beta )(p) = \alpha (p)$ for $p\in \dom \alpha $ and $(\alpha \cup \beta )(p) = \beta (p)$ for $p\in \dom \beta $.
If $\alpha$ and $\beta$ are nonoverlapping, the \emph{cut of the union $\alpha \cup \beta $} is the set of edges of $L_{\square }$ with one end-point in $\dom \alpha $ and the other end-point in $\dom \beta$.
An assembly $\gamma $ is \emph{singular} if $|\dom \gamma| = 1$.
We say that an assembly $\gamma $ is \emph{$\mathcal{T}$-producible} if either $\gamma $ is singular or there exist $\mathcal{T}$-producible nonoverlapping assemblies $\alpha$ and $\beta$ such that $\gamma = \alpha \cup \beta$ and
the cut of $\alpha \cup \beta $ is $\tau$-stable.
In the latter case, we write $\alpha + \beta \to_{1}^{\mathcal{T}} \gamma$.
Note that every $\mathcal{T}$-producible assembly is $\tau$-stable.
A $\mathcal{T}$-producible assembly $\alpha$ is \emph{$\mathcal{T}$-terminal} if there are no $\mathcal{T}$-producible assemblies $\beta$ and $\gamma$ such that $\alpha + \beta
\to_{1}^{\mathcal{T}} \gamma$.
We say two assemblies $\alpha$ and $\beta$ are \emph{equivalent up to translation}, written $\alpha \simeq \beta $, if there is a vector $\vec x\in \ZZ^{2}$ such that $\dom \alpha = \dom \beta + \vec x$ and for all $p\in \dom \beta$, $\alpha (p + \vec x) = \beta (p)$.
We say that $\mathcal{T}$ \emph{uniquely produces} $\alpha$ if $\alpha$ is $\mathcal{T}$-terminal and for every $\mathcal{T}$-terminal assembly $\beta$, $\alpha \simeq \beta$.

A \emph{restriction} of an assembly $\alpha$ to a set $D\subseteq \dom \alpha$, denoted by $\alpha{\restriction_{D}}$, is $\dom \alpha{\restriction_{D}} = D$ and for every $p \in D$, $\alpha{\restriction_{D}}(p) = \alpha (p)$. 
If $C$ is a subgraph of $L_{\square }$ such that $V(C)\subseteq \dom \alpha $, we define $\alpha{\restriction_{C}} = \alpha{\restriction_{V(C)}}$.

When $\calT$ is clear from context, we may omit $\calT$ from the notation above and instead write
$\to_1$,
$\to$,
\emph{produces},
\emph{producible}, and
\emph{terminal}.


\end{document}